\documentclass[submission,copyright,creativecommons]{eptcs}
\providecommand{\event}{AFL 2017}

\usepackage{underscore}

\usepackage{todonotes}
\usepackage{microtype}

\usepackage{amsmath}
\usepackage{amsthm}
\usepackage{adjustbox}
\usepackage{amssymb}
\usepackage[utf8]{inputenc}
\usepackage{tikz}
\usepackage[normalem]{ulem}    

\usepackage{hyperref}

\newtheorem{theorem}{Theorem}[section]
\newtheorem{corollary}[theorem]{Corollary}

\newtheorem{lemma}[theorem]{Lemma}

\theoremstyle{definition}
\newtheorem{definition}[theorem]{Definition}
\newtheorem{example}[theorem]{Example}

\usepackage{tikz}		
\usetikzlibrary{automata}	
\usetikzlibrary{decorations.pathmorphing,arrows,positioning}
\usetikzlibrary{calc}
\usepackage{tikzscale}
\usetikzlibrary{shapes.geometric}
\usetikzlibrary{arrows,automata}
\usetikzlibrary{matrix}
\usetikzlibrary{decorations.markings}
\usetikzlibrary{decorations.shapes}

\DeclareMathOperator{\Red}{Red}
\DeclareMathOperator{\Loc}{SLT}

\newcommand{\tagged}{\overline{\Red}}
\newcommand{\taggable}{\text{HOP}{}}
\newcommand{\maxtag}{\Red}

\newcommand{\ftagged}{\tagged}
\newcommand{\ftaggable}{\taggable}
\newcommand{\fmaxtag}{\maxtag}

\newcommand{\ltagged}{\overline{\Red}}
\newcommand{\ltaggable}{\Red}

\newcommand{\csharp}{\text{\textcircled{\#}}}

\title{Higher-Order Operator Precedence Languages}

\author{Stefano {Crespi Reghizzi}
	\institute{DEIB, Politecnico di Milano, and\\
    IEIIT, Consiglio Nazionale 
	delle Ricerche\\
		via Ponzio 34/5, 20134 Milano, Italy}
	\email{stefano.crespireghizzi@polimi.it}
\and
Matteo Pradella
\institute{DEIB, Politecnico di Milano, and\\ 
IEIIT, Consiglio Nazionale delle 
Ricerche\\
via Ponzio 34/5, 20134 Milano, Italy}
\email{matteo.pradella@polimi.it}
}
\def\authorrunning{S. {Crespi Reghizzi} \& M. Pradella}
\def\titlerunning{Higher-Order Operator Precedence Grammars}

\begin{document}

\maketitle

\begin{abstract}
Floyd's Operator Precedence (OP) languages are a deterministic context-free
family having many desirable properties.  They are locally and parallely
parsable, and languages having a compatible structure are closed under Boolean
operations, concatenation and star; they properly include the family of Visibly Pushdown (or Input
Driven) languages.  OP languages are based on three
relations between any two consecutive terminal symbols, which assign syntax
structure to words.  We extend such relations to $k$-tuples of consecutive terminal
symbols, by using the model of strictly locally testable regular languages of
order $k \ge 3$.  The new corresponding class of {\em Higher-order
Operator Precedence languages} (HOP) properly includes the OP languages, and it is
still included in the deterministic (also in reverse) context free family. 
We prove Boolean closure for each subfamily of structurally compatible HOP languages.
In each subfamily, the top language is called {\em max-language}. We show that 
such languages are defined by a simple cancellation rule and we  
prove several properties, in particular that max-languages make an infinite hierarchy ordered by parameter $k$.
HOP languages are a candidate for replacing OP languages in the various applications where 
they have have been successful though sometimes too restrictive.

\end{abstract}


\section{Introduction}\label{sect:introduction}
\par
We propose a new way of extending the classic language family of \emph{operator-precedence} (OP) languages, invented by R. Floyd \cite{Floyd1963} to design a very efficient parsing algorithm,  still used within compilers.
It is worth outlining the main characteristics of OP languages.
OP languages have been also exploited for 
grammar inference \cite{journals/csur/AngluinS83}, thanks to their 
lattice-theoretical properties.  They offer promise for model-checking of 
infinite-state systems 
due to the Boolean closure, $\omega$-languages, logic and automata characterizations, and the ensuing decidability of 
relevant problems \cite{LonatiEtAl2015}. Recently, a generator of fast parallel parsers
 has been made available \cite{BarenghiEtAl2015}.
Their bottom-up deterministic parser localizes the edges of the handle (a 
factor to be reduced by a grammar rule)  by means of three precedence 
relations, represented by the \emph{tags}  $\lessdot, \gtrdot, \dot=$.
(Since our model generalizes OP, we represent the tags as 
$[, ], \odot$.) Such 
relations are defined  between two consecutive terminals (possibly separated by a 
nonterminal). E.g., the \emph{yield precedence} relation $a \lessdot b$ says 
that $b$ is the leftmost terminal of the handle and $a$ is the last terminal of 
the left context. The no-conflict condition of OP grammars ensures that the edge 
positions are unambiguous and the handles can be localized by means of a local 
test. 
An OP parser configuration is essentially  a word consisting of alternated 
terminals and tags, i.e., a \emph{tagged word}; notice that nonterminal 
symbols, although present in the configuration, play no role in determining the 
handle positions, but are of course necessary for  checking syntactic 
correctness.
In general, any language having the property that handles can be localized by a local test is  called \emph{locally parsable} and its parser is amenable to parallelization.
\par
If the parser is abstracted as a pushdown automaton, each pair of terminals  
associated to a left or to a right edge of a handle, respectively triggers a 
push or a pop move; i.e.,  the move choice is driven by two consecutive input 
symbols. Therefore, the well-known model  of \emph{input-driven} 
\cite{DBLP:conf/icalp/Mehlhorn80,Input-driven} (or ``visibly pushdown'' 
\cite{jacm/AlurM09}) languages is a special case of the OP model, since just 
one  terminal suffices to choose the move. This is shown in   
\cite{CrespiMandrioli12}, where the precedence relations characterizing  the 
input-driven languages are computed.
 The syntax structures permitted by such relations are sometimes too restrictive for the constructs of modern languages, e.g., a markup language like HTML5 has special rules that allow dropping some closing tags.
\par
Since OP grammars, although used by compilers, are sometimes inconvenient or inadequate for specifying some syntactic constructs, a natural question is: can we increase the generative capacity of OP grammars, without jeopardizing their nice properties, by allowing the parser to examine more than two consecutive  terminals to determine the handle position? 
Quite surprisingly, to our knowledge the  question remained  unanswered until now, but in the last section we mention some related research.
\par
We intuitively present the main ideas of the new hierarchical family of languages and grammars called \emph{Higher-order Operator Precedence} (HOP). Let $k\geq 3$ be and odd integer specifying the number of consecutive terminals and intervening tags to be used for localizing handles: the value of $k$ is 3 for OP, which thus coincide with the HOP(3) subfamily. 
The main contributions are: a precise definition of HOP($k$) grammars,  a decidable condition for testing whether a grammar has the  HOP($k$) property, the proof that the OP family is properly included into the 
HOP one, and an initial set of nice properties that carry over from OP to HOP. 
The Boolean closure of each  structurally compatible (this concept cannot be 
defined at this point but is standard for OP and input-driven languages) HOP 
subfamily is determinant for model checking.
Concerning local parsability, we mention in the conclusions how it should be 
obtained. 
Last but not least, our definition of HOP grammars permits to use regular expressions 
in the right part of rules, in contrast with the classical definition of OP grammars.

Moreover, we prove that each structurally compatible HOP subfamily  has a maximal element, called max-language.
Interestingly, max-languages can be defined by a simple cancellation rule that applies to tagged words, and iteratively deletes innermost handles by a process called a \emph{reduction}. Before each cancellation, the word, completed with tags, has to pass local tests, defined by means of a \emph{strictly locally testable} \cite{McNaughtonPapert71} regular language of order $k$.  We prove several properties of the max-language family, in particular that they form a strict infinite hierarchy ordered by parameter $k$.
Since the model based on cancellation is simpler, it will be the first  presented in this paper, before the HOP grammar model.

 \par
Paper organization:
Section~\ref{sect:basicDef} contains the basic notation and definitions. 
Section~\ref{sect:SLRlanguage} introduces  the max-languages and their basic properties. 
 Section~\ref{sect:generalizationOpPrecLang}   defines the HOP grammars and proves their properties. 
Section~\ref{sect:RelatedWorkConcl} compares HOP with some related existing 
models, and  lists open problems and  future research directions.

\section{Basic definitions}\label{sect:basicDef}
For terms not defined here, we refer to any textbook on formal languages, e.g. 
\cite{Harrison78}. 
For a generic alphabet we use the symbol $\Upsilon$. 
The empty word is denoted by $\varepsilon$. Unless stated otherwise, all 
languages considered are free from the empty word. 
For any $k \ge 1$, for a word $w$, $|w|\ge k$, let $i_k(w)$ and $t_k(w)$ be the prefix and, 
respectively, the suffix of $w$ of length $k$.
If a word $w$ has length at least $k$, $f_k(w)$ denotes the set of 
factors of $w$ of length $k$, otherwise the empty set. Obviously, $i_k(w)$, 
$t_k(w)$  and  $f_k$ can be extended to languages. The $i$-th 
character of $w$ is denoted by $w(i), 1 \le i \le |w|$.
 
\par
A (nondeterministic) \emph{finite automaton} (FA) is denoted by 
$M=(\Upsilon,Q,  \delta, I, T)$, where  $I, T \subseteq Q$  are respectively 
the initial and final states and $\delta$ is a relation (or its graph) over $Q 
\times \Upsilon \times Q$.
A (labeled) \emph{path}  is a sequence  
$q_1\stackrel {a_1} \to q_2 \stackrel {a_2} \to \dots \stackrel {a_{n-1}} \to q_n$,
such that, for each $1 \le i < n$,  $(q_i, a, q_{i+1}) \in \delta$. The {\em path label} is $a_1 a_2 \dots a_{n-1}$,
the {\em path states} are the sequence $q_1 q_2 \ldots q_n$.
An FA is \emph{unambiguous} if  each sentence in $L(M)$  is recognized by  just one computation.
\par
  An \emph{extended context-free} (ECF) grammar is a 4-tuple $G=(V_N,\Upsilon, 
  P, 
  S)$, where $\Upsilon$ is the terminal alphabet, $V_N$ is the nonterminal 
  alphabet, $P$ is  the set of rules, and  
  $S\subseteq V_N$ is the set of axioms.  
	Each rule has the form $X \to R_X$, where $X \in V_N$ and $R_X$ 
  is a \emph{regular} language  over the alphabet $V = \Upsilon \cup V_N$. 
  $R_X$ will be  defined 
  by means of an unambiguous 
   FA,  $M_X=(V, Q_X, \delta_X, I_X , T_X)$. We safely assume that for each nonterminal $X$ there is exactly one rule, to be written as $X \to M_X$ or $X \to R_X$.
   A rule $X \to R_X$ is a {\em copy rule} if  $\exists Y \in V_N: Y \in R_X$;
  we  assume that there are no copy rules.
A \emph{context-free} (CF) grammar is an ECF grammar such that for each rule $X \to R_X$, $R_X$ is a finite language over $V$.
\par
The  \emph{derivation relation} $\Rightarrow\, \subseteq V^* \times  V^*$ is defined as follows for an ECF grammar:
 $u  \Rightarrow v$ if $u = u' X u''$,  $v = u' w u''$, $X \to R_X\in P$, and 
 $w \in R_X$.
\par
A word is $X$-\emph{grammatical} if it derives from a nonterminal symbol $X$. 
If $X$ is an axiom, the word is \emph{sentential}. The language generated  by 
$G$ starting from a nonterminal $X$ is denoted by $L(G,X) \subseteq \Upsilon^+$ 
and   $L(G) = 
\bigcup_{X  \in S}{L(G,X)}$.
\par
The usual assumption that all parts of a CF grammar are productive  can be reformulated for ECF grammars by combining reduction (as in a CF grammar) and trimming of the $M_X$ FA for each rule $X \to M_X$, but we omit details for brevity.

\par
An ECF grammar is in \emph{operator} (normal) \emph{form} if for all rules 
$X\to R_X$ and for each  $x \in R_X$, $f_2(x) \cap V_N V_N=\emptyset$, i.e. it 
is impossible to find two adjacent nonterminals.
Throughout the paper we only consider ECF grammars in operator form.


\par 
Let $G=(V_N, \Upsilon, P, S)$ and assume that  $\{(, )\} \cap \Upsilon
= \emptyset$. The \emph{parenthesis grammar} $G_{(\,)}$ is defined by the
4-tuple $(V_N, \Upsilon \cup \{(, )\}, P', S)$  where  $P'= \{X \to (R_X) \mid 
X 
\to R_X \in P\}$. 
Let $\sigma'$  be the  homomorphism  which 
erases parentheses, a grammar $G$ is \emph{structurally
ambiguous} if there exist $w,z \in L(G_{()}), w \ne z$, such that $\sigma'(w) = 
\sigma'(z)$.
Two grammars $G'$ and $G''$ are
\emph{structurally equivalent} if $L(G'_{(\,)})=L(G''_{(\,)})$.

\medskip
\noindent
{\bf Strict local testability and tagged languages }
Words of length $k$ are  called $k$-\emph{words}. 
The following definition, equivalent to the classical ones (e.g., in 
\cite{McNaughtonPapert71},\cite{DBLP:journals/tcs/Caron00a}), assumes that any 
input word $x \in \Upsilon^+$ is enclosed between two special words of 
sufficient length, called 
{\em end-words} and 
denoted by $\csharp$.
Let $\#$ be a character, tacitly assumed to be in $\Upsilon$ and used only in 
the end-words. We actually use two different end-words, without or with tags, 
depending on the 
context: $\csharp \in \#^+$ (e.g. in Definition~\ref{defSLT})
or 
$\csharp \in (\# \odot )^* \#$, (e.g. in 
Definition~\ref{defLSTlist}).

\begin{definition}\label{defSLT}
Let $k\geq 2$ be an integer, called \emph{width}. A language $L$  is 
$k$-\emph{strictly locally testable}, if there exists a 
$k$-word set $F_k \subseteq \Upsilon^k$ 
such that
$ L= \{x \in \Upsilon^* \mid f_k\left(\csharp\, x \,  \csharp\right)   
\subseteq F_k \} $;
then we write $L= \Loc(F_k)$.
A language is \emph{strictly locally testable} (\emph{SLT}) if it is $k$-strictly locally testable for some $k$. 
\qed
\end{definition}

We assume that the three characters, called \emph{tags},  $[, ]$, and $\odot$  
are distinct from terminals and nonterminals characters and we denote them as 
$\Delta =\{ [, ], \odot 
\}$.
For any alphabet, the projection $\sigma$ erases all the tags, i.e. 
$\sigma(x) = \varepsilon$, if $x \in \Delta$, otherwise $\sigma(x) = x$.
Here we apply the SLT definition to words that contain tags and are 
the base of our models.
Let $\Sigma$ be the terminal alphabet. 
A \emph{tagged word} starts with a terminal and
alternates tags and terminals.

\begin{definition}[tagged word and tagged language]\label{defLSTlist}
Let here and throughout the paper $k\geq 3$ be an odd integer.
A \emph{tagged word}  is a word   $w$ in the set $\Sigma(\Delta\Sigma)^*$, denoted by $\Sigma^\square$. 
A \emph{tagged sub-word} of $w$ is a factor of $w$ that is a tagged word.
A \emph{tagged language} is a set of tagged words.
Let $\Sigma^{\square k}=\{w \in \Sigma^\square \mid |w| =k \}$. We call \emph{tagged} $k$-\emph{word} any word in $\Sigma^{\square k}$.
The set of all tagged $k$-words that occur in $w$ is denoted by $\varphi_k(w)$.
\par
A  language $L\subseteq \Sigma^\square$  is a 
$k$-\emph{strictly locally testable tagged language} 
if there exists a set of tagged $k$-words $\Phi_k \subseteq \Sigma^{\square k}$  such that
$L= \left\{w \in \Sigma^\square \mid \varphi_k\left(\csharp \, [ \, \, w \, ] \, \csharp  \right)  \subseteq \Phi_k \right\}. $ 
In that case we write $L=\Loc(\Phi_k)$. 
The $k$-word set  $F_k\subseteq (\Sigma \cup \Delta)^k$ 
\emph{derived from}  $\Phi_k$ is $F_k= \bigcup_{x \in \Loc(\Phi_k)}  f_k(x)$.
\par
A tagged $k$-word set $\Phi_k$ is \emph{conflictual} if, and only if, $\exists 
x, y \in \Phi_k, x \ne y$, such that $\sigma(x) = \sigma(y)$.
\qed
\end{definition}
\par
 E.g., $\Loc(\{\#  [  a, \  a\odot b,\  b \odot a ,\  a] \#\}) = (a \odot b\,\odot)^* a$.

We  observe that, for each word $w\in\Sigma^\square$, the set  $\varphi_k(w)$ is included in $f_k(w)$.
E.g., from $\Phi_3= \{\#  [ a,\  a\odot b,\  b \odot a,\  a]\#\}$ we derive the 3-word set $F_3=\Phi_3 \cup \{ [  a\odot,\  [ a ],\  \odot b \odot,\  \odot a \odot,\ \odot a ]\}$.
Yet, although $\Phi_k\subset F_k$, the languages defined by strict local 
testing obviously coincide:  $\Loc(F_k) = \Loc(\Phi_k)$.

In what follows all tagged word sets considered are not conflictual, unless 
stated otherwise.
An important remark is that for every word $w$ over $\Sigma$, $\sigma^{-1}(w) 
 \cap \Loc(\Phi_k)$ is either empty or a singleton: the tagged word 
 corresponding to $w$.
\par
The following technical lemma is useful for later proofs.
\begin{lemma}\label{lemmaSLTlistLang}
	Let $w \in \Sigma^{\square k}$; let  $s', s'' \in \Delta$ be two distinct tags.
	Then,
    for 
	every $3 \le h \le k+2$,
    the tagged word set $\varphi_{h} (w s' w s'' w)$ is conflictual.
	
\end{lemma}
\begin{proof} 
	Let $w= a_1 s_2 a_3 \ldots s_{k-1} a_k$. It suffices to observe that the 
	conflicting tagged $h$-words 
	$t_h(a_1 s_2 a_3 \ldots  s_{k-1}$ $ a_k s'  a_1)$   and 
	$t_h(a_1 s_2 a_3 \ldots  s_{k-1} a_k s'' a_1)$
	are contained in $\varphi_{h} (w s' w s'' w)$.
\end{proof}
An immediate corollary: when $w = a\in \Sigma$, for any sufficiently long word  
$z \in a (\Delta a)^*$, if $z$ contains two distinct tags,
the set $\varphi_k(z)$ is conflictual.

\section{Reductions and maximal languages}\label{sect:SLRlanguage}
We show that the SLT tagged words, defined by a set $\Phi$ of 
(non-conflictual) $k$-words, can be interpreted as defining another language over 
the terminal alphabet; the language is context-free but not necessarily regular, and 
is called 
{\em maximal language or max-language}. 
We anticipate from Section~\ref{sect:generalizationOpPrecLang} the reason of the 
name 
``max-language'': such languages   belong to the family of 
Higher-order Operator Precedence languages (Definition \ref{defHOP}), and they 
include 
any  other HOP language that is structurally compatible. 
We first define the reduction process, then we prove some properties of the 
language family.

Consider a set $\Phi_k \subseteq \Sigma^{\square k}$ and a word $w$ over $\Sigma$; let $x \in \Loc(\Phi)$ be the tagged word 
corresponding to $w$, if it exists. 
Word $w$ belongs to the max-language if $x$  reduces to a 
specified word by the repeated application of a reduction operation. 
A reduction cancels from the current  $x$  a sub-word  of a special form 
called a handle, replaces it  with a tag,  and thus produces a new tagged word. All the  
tagged words thus obtained by successive reductions must belong to $\Loc(\Phi)$.

\begin{definition}[maximal 
language]\label{defHandle}\label{defStrongReduction}\label{defSLRlanguage}
Let   $\Phi \subseteq \Sigma^{\square k}$.
A \emph{handle} is a word of the form  $ [ \,  x ] \, $ where $x \in 
\left(\Sigma - \{\#  \}\right) \cdot \left(\odot \left(\Sigma  - \{\#  
\}\right)\right)^* $, 
i.e.,   a handle is a tagged word enclosed between the tags [ and ], and not containing symbols in $\left\{[,],\# \right\}$.

\noindent
A \emph{reduction} is a binary relation  $\leadsto_{\Phi}  \subseteq  
(\Sigma \cup \Delta)^+ \times (\Sigma \cup \Delta)^+$ between 
tagged words, defined as:
\begin{equation}\label{eqDefStrongReduction1}
w [u] z \leadsto_{\Phi} w s z \text{ if, and only if, } w [u] z 
\in 
\Loc(\Phi) \text{ where } [u] 
\text{ is a handle, and }  \exists s \in 
\Delta:
 w s z \in \Loc(\Phi). 
\end{equation}
The  handle $[u]$ is called \emph{reducible}.
A reduction is called \emph{leftmost} if no handle occurs in $w$.
The definition of \emph{rightmost} reduction is similar.
Observe that at most one tag $s$ may satisfy 
Condition~\eqref{eqDefStrongReduction1} since  $\Phi$ is non-conflictual.
The subscript  $\Phi$ may be dropped from $ \leadsto_\Phi$ when clear from context;
$\stackrel* \leadsto$ is the reflexive and  transitive closure of $\leadsto$.
%

\noindent 
The \emph{tagged maximal language}  defined by $\Phi$ via {\em reduction} is 
$
\ltagged(\Phi) = 
\{ w \in \Sigma^\square \,\mid\; \csharp \,  [ \, \,w\,] \csharp 
\;\stackrel * \leadsto_{\Phi} \,\csharp\odot \csharp \}.
$
The \emph{maximal language} defined by  $\Phi$, is   
$\ltaggable(\Phi) = \sigma\left( \ltagged(\Phi)\right)$.
\par
We say that languages  $\ltaggable(\Phi)$ and $\ltagged(\Phi)$ are in the families 
$\fmaxtag(k)$ and $\ftagged(k)$ respectively;   a language is in the
$\fmaxtag$ family  if it is in $\fmaxtag(k)$ for some $k$.
\qed
\end{definition}
Notice that in  Definition~\ref{defSLRlanguage} the reductions may be applied in any 
order, 
without affecting $\ltagged(\Phi)$.

\begin{example}\label{ex:Dyck}
	This and the following examples were checked by a  program.
	The Dyck language (without $\varepsilon$) over the alphabet $\{a, a', b, 
	b'\}$, which can be easily extended to an arbitrary number of matching 
	pairs,
	 is a $\Red(3)$ language defined by the tagged word set
	$\Phi = \{ 
	\# \odot \# ,\ 
	b']a',\ 
	a']b',\ 
	\#[b,\ 
	b'[b,\ 
	b[b,$\ 
	$a']a',\ 
	\#[a,\ 
	b']\#,\ 
	a \odot a',\ $
	$a[b, \ 
	b[a, \ 
	b'[a,\ 
	a'[b,\ 
	b']b',\ 
	a']\#,$\ 
	$b \odot b',\ 
	a'[a,\ 
	a[a  
	\}.
	$
	Word  $a aa' a' aa'$ is recognized by the following  reductions, 
	respectively leftmost and rightmost: 
	\[
	\begin{array}{l|l} 
	\csharp [a[a \odot a']a'[a \odot a']\csharp   \leadsto
	\csharp [a \odot a'[a \odot a']\csharp   
	&
	\csharp [a[a \odot a']a'[a \odot a']\csharp   \leadsto
	\csharp [a[a \odot a']a']\csharp    \\ 
	\leadsto \csharp [a \odot a']\csharp   \leadsto \csharp  \odot \csharp    &
	\leadsto \csharp [a \odot a']\csharp   \leadsto \csharp  \odot \csharp   \\ 
	\end{array}
	\] 
\end{example}

Some elementary properties of max-languages come next.
\begin{lemma}\label{lemmaSLRproperties}
\begin{enumerate}
    \item $\nexists x,y \in \ltagged(\Phi)$, $x \ne y$: $\sigma(x) = 
    \sigma(y)$. (Unambiguity)
	\item $\forall x, y$, if $x \in \ltagged(\Phi)$ and 
  $\csharp  \,  [ \, \,x\,] \,  \csharp  \;\stackrel * \leadsto_{\Phi} \,\csharp  \,  [ \, \,  y \,] \,  \csharp  $, 
then 
  $y \in \ltagged(\Phi)$. (Closure under $\leadsto$)
	
	\item  Let $F_h= \sigma(\Phi)$ (hence $h= \lceil k/2 \rceil$). Then   
$
	\Loc(F_h) \supseteq \ltaggable(\Phi).$ (Refinement over SLT)
\end{enumerate}

\end{lemma}
\begin{proof}
Stat. 1. and 2. follow from Definition \ref{defSLRlanguage}. 
Define the tagged $k$-words set  $\hat{\Phi}= \sigma^{-1}(F_h)\cap 
\Sigma^{\square k}$, which clearly includes $\Phi$.  From the identity 
$\Loc(F_h)= \sigma(\Loc(\hat{\Phi}))$, Stat.~3 follows. 
\end{proof}

\begin{example}\label{ex:FirstExample} This is a running example. 
	$L = \{a^n (c b^+)^n \mid n > 0\}$
is a $\maxtag(3)$ language specified by
$
\Phi= \{ 
\# \odot \#,\ 
\# [ a,\ 
b ] \#,\ $
$b ] c,\ 
c \odot b,\ 
b \odot b,\ 
a \odot c,\ 
a [ a
\}
$.
The reduction steps for  word  $aaacbbbbcbcbbb$ are:
$ \begin{array}{l}
    \csharp [a[a [a \odot c \odot b \odot b \odot b \odot b] c \odot b]c \odot b \odot b \odot b]\csharp   \leadsto \\
\csharp [a[a \odot c \odot b]c \odot b \odot b \odot b]\csharp    \leadsto 
\csharp [a \odot c \odot b \odot b \odot b]\csharp   \leadsto 
\csharp  \odot \csharp   
\end{array} 
$\\
Therefore $ [a[a[a \odot c \odot b \odot b \odot b \odot b]c \odot b]c \odot b \odot b \odot b]  \in \ltagged(\Phi)$ and 
$aaacbbbbcbcbbb  \in  \ltaggable(\Phi)$.
On the other hand, word $aa cbb$ is not accepted because it is not the content of a tagged word that reduces to $ \csharp  \, \odot \, \csharp  $, 
in particular, the reduction of handle $ [a \odot c \odot b \odot b]$
in $\csharp [a [a \odot c \odot b \odot b] \csharp$
is not possible because there is not a tag $s$ such that $ a s \#  \in  
\Loc(\Phi)$.

\end{example}

%


First, we compare $\fmaxtag$ with REG, the family of {\em regular languages}.

\begin{theorem}\label{theorSLTinclSLRincompREG}
 The family of $\fmaxtag$ languages strictly includes the SLT family and is incomparable 
 with the REG family. 
\end{theorem}
\begin{proof}
Inclusion SLT $\subseteq \fmaxtag$: 
The mapping $\widetilde{(\cdot)} : \Sigma^+ \to \Sigma^{\square}$ is defined by 
$\widetilde{z}= z(1) \odot z(2) \odot \dots \odot z(|z|) $, for any $z \in \Sigma^+$. 
Given a set $F_j$, $j\ge 2$ defining  an SLT language over $\Sigma$, we define the set $\Phi_k$, $k= 2j-1$: it contains, for every $u \in F_j \cap \left(\Sigma - \{\# \}\right)^+$, the tagged word
$\widetilde{u}$, and, for every 
$w= \#  ^{j_1} v  \#  ^{j_2} \in F_j$, $j_1 + |v| + j_2 = j$, 
the tagged word  
$\widetilde{w}=(\#\odot  )^{j_1 -1} \# [ \,  v(1) \odot v(2) \odot \dots \odot v(|v|) ] \,\#  (  \odot\#)^{j_2-1}$. 

\noindent We prove  that  $\Loc(F_j) \subseteq \ltaggable(\Phi_k)$. 
Consider any  $z\in \Loc(F_j)$, for simplicity assume $|z| \geq j$. 
Then   $\widetilde{\#^{j-1}   z \#^{j-1} } = \csharp   [z(1) \odot \dots \odot z(|z|) ]  \csharp$  
$\leadsto_{\Phi_k} \,  \csharp   \odot \csharp   $.   Since the converse 
inclusion $\Loc(F_j)\supseteq \ltaggable(\Phi_k)$ is   obvious, it follows that 
SLT $\subseteq \fmaxtag$.  

\par
The inclusion SLT $\subset \fmaxtag$  is proved by using
$L=a^* b a^* \cup a^+$ which is defined by 
$\Phi_3 = 
\{
\#[b,
a]b,
\# \odot \#,
\# [ a,
b]\#,
a \odot a,
b[a,
a]\#
\}
$. 
But it is known that $L$ is not locally testable.

\par
The inclusion $\fmaxtag \not\subseteq$ REG is proved by the Dyck languages.
To prove REG $\not\subseteq \fmaxtag$, we consider $R=(aa)^+$. 
By Lemma~\ref{lemmaSLTlistLang}, a $\Phi_k$ for $R$ may only use one tag, and we first consider the case [ (the case with ] is analogous). For any odd value $k \geq 3$,  $\Phi_k$ has the form
$\{
\csharp [ a, \ldots
\#  [ a [ \ldots  [  a, 
a [ \ldots  [  a, 
\ldots
a [ \ldots  [ a ] \csharp,
\ldots
\}$, 
therefore the  handle is always $[a]$ and $\ltaggable(\Phi_k)$ necessarily includes also words with an odd number of $a$'s. 
The same conclusion holds in the $\odot$ case, since  any handle has the form $[ a\odot \ldots \odot a ]$. 
\end{proof}

We prove that family $\fmaxtag$ is an infinite strict hierarchy.
\begin{theorem}\label{theorSLRk<SLRk+1}
For every $k\geq 5$, the language family  $\fmaxtag(k)$ strictly includes $\fmaxtag(k-2)$.
\end{theorem}
\begin{proof}
Consider the languages $L_{(h)} = (a^h b)^+$, $a \ge 1$. 
It is easy to verify  that for $k = 2h+1$, 
$L_{(h)}$ is in $\fmaxtag(k)$. 
E.g., $L_{(2)} = \ltaggable(\{ 
\# \odot \#[a,\ 
b]\# \odot \#,\ 
b]a \odot a,\ 
a \odot b]a,\ 
a \odot a \odot b,\ 
\#[a \odot a,\ 
a \odot b]\#,\ 
\# \odot \# \odot \#
\})$.

We prove that $L_{(h)}$ is not in  $\fmaxtag(k-2)$.
Assume by contradiction that  $L_{(h)}=\ltaggable(\Phi_{k-2})$, for some $\Phi_{k-2}$, and consider $y \in L_{(h)}$ and $y'\in \ltagged(\Phi_{k-2})$, such that $y=\sigma(y')$. 
The word $y'$ contains a
  tagged sub-word $w= a s_1 \ldots a s_{h-1} a$, $s_i \in \Delta$, and two cases are possible.
\\ \noindent $\bullet$ $\exists 1 \leq i < j \leq h-1$ such that $s_i \neq s_j$. By Lemma \ref{lemmaSLTlistLang}, the set $\varphi_{k-2}(w)$ is conflictual.
\\ \noindent $\bullet$ All $s_i$ in $w$ are identical.
  	But this means that, if $a^h b \in \ltaggable(\Phi_{k-2})$ (by hypothesis), then also $a^{h+1} b \in \ltaggable(\Phi_{k-2})$, which is not in $L_{(h)}$.
\end{proof}

\par
The $\fmaxtag$ family is not closed under the following operations, as proved by
witnesses:\\
\noindent {\bf Intersection}: 
$\{a^n b^n c^* \mid n > 0\} = \ltaggable(\{
b[c,\ 
c]\#,\ 
a \odot b,\ 
c \odot c,\ 
\# \odot \#,\ 
\#[a,\ 
b]\#,\ 
b]b,\ 
a[a\ 
\})$ and
$\{a^* b^n c^n \mid n > 0\} =
\ltaggable(\{
\#[b,\ 
a]b,\ 
c]\#,\ 
c]c,\ 
\# \odot \#,\ 
b[b,\ 
\#[a,\ 
a \odot a,\ 
b \odot c
\})$, and their intersection is not  context-free.
\\
\noindent {\bf Set difference}:
$(a^* b a^* \cup a^+) - a^+$.
The first language is in Red(3) (see the proof of Theorem~\ref{theorSLTinclSLRincompREG}) and requires $\csharp [ a$ and $a ] \csharp$ because
words may begin and end with $a$. Since 
unlimited runs of $a$ are possible, Lemma~\ref{lemmaSLTlistLang}
imposes the same tag between any pair of  $a$'s, hence
the resulting $\maxtag$ language necessarily contains  $a^+$,  a contradiction.\\
\noindent {\bf Concatenation}:
$a^* b = \ltaggable(\{
\#[b,\ 
a\odot b,\ 
\#\odot\#,\ 
\#[a,\ 
b]\#,\ 
a\odot a
\})$
concatenated with 
$a^+$ is similar to the witness for set difference.\\
\noindent {\bf Intersection with regular set}:
$\{a, b\}^+ \cap (a a)^+$, see Theorem~\ref{theorSLTinclSLRincompREG}.

The language family of the next section contains $\fmaxtag$ and has better 
closure properties.

\section{Generalization of operator-precedence languages}\label{sect:generalizationOpPrecLang}

We introduce a new family of languages, called $\taggable(k)$, standing for 
{\em Higher-order Operator Precedence languages} of order $k\geq 3$. The 
$\taggable(k)$ condition is decidable for grammars; $\taggable(k)$ languages are deterministic, also in reverse. With respect to existing families, we start from the operator-precedence (OP) 
languages, defined by Floyd \cite{Floyd1963}, and prove that they are the same 
as the new family 
$\ftaggable(3)$, when the grammar rules are in non-extended CF form.  We prove 
the Boolean closure property  for  the family of $\taggable(k)$ languages 
having the same set $\Phi_k$ of  tagged $k$-words. The top element in such 
family is the $\maxtag$ language (also known as max-language) defined by $\Phi_k$. 

\medskip
\noindent {\bf Operator Precedence Grammars }
An Operator Precedence grammar\footnote{Floyd's definition uses CF 
grammars, but it is straightforward to extend it to ECF grammars.}
\cite{Floyd1963}
is characterized by three OP relations over $\Sigma^2$, denoted by $\gtrdot$, $\doteq$, $\lessdot$, that are used to assign a structure to the words (see e.g. \cite{GruneJacobs:08} for the classical parsing algorithm).

\begin{example}\label{ex:opgram}
Consider the following operator grammar (with rules specified for brevity by regular expressions) and its OP relations:
\[ 
\begin{array}{lcl}
G_1 =  
\left \{
S \to X b X \cup b X, \ 
X \to a a^*
\right\}
&
\ \ \ 
&
a \doteq a,\, a \gtrdot b,\, b \lessdot a.
\end{array}
\]
By default, $\# \lessdot x$ and $x \gtrdot \#$, for any $x \in \Sigma$, and $\# 
\doteq \#$.
A grammar is OP if at most one relation exists between any two terminal symbols.
To parse word $aaaba$, the bottom-up OP parser executes the following 
reduction steps:
\begin{equation}
\# \lessdot a \doteq a \doteq a \gtrdot b \lessdot a \gtrdot \#  \Longleftarrow_{G_1}
\# \stackrel{X}{\lessdot} b \lessdot a \gtrdot \# \Longleftarrow_{G_1}
\# \stackrel{X}{\lessdot} b \stackrel{X}{\gtrdot} \# \Longleftarrow_{G_1}
\# \stackrel{S}{\doteq} \# 
\label{eqOPreduction}
\end{equation}
If we substitute the OP relation symbols with tags, the above 
OP relations are encoded by the tagged 3-words $\Phi_3 = \{ a \odot a, a [ b, b ] a, \# [ a, a ] \#, \# \odot \#, \# [ b, b ] \# \}$. 
The OP property implies that $\Phi_3$ is non-conflictual.  
Notice that each word in \eqref{eqOPreduction} (disregarding the \#'s) belongs to $\ltagged(\Phi_3)$.
Recalling Definition~\ref{defSLRlanguage}, we observe that $a\odot a \odot a ] b [ a \in \ltagged(\Phi_3)$, therefore $aaaba \in \ltaggable(\Phi_3)$. 
Moreover, $\ltaggable(\Phi_3) = a^* b a^* \cup a^+ \supset L(G_1)$. 
\end{example}

Our generalization of OP grammars is based on the idea of using tagged 
$k$-words, with $k \ge 3$, for assigning a syntactic structure to words.
The test for a grammar to be OP \cite{Floyd1963} is  quite simple, but  the wider contexts  needed when $k>3$, impose 
a more involved  device for the no-conflict check. For that we define a grammar, called {\em tagged}, obtained from
the original grammar by inserting tags into rules.

\begin{definition}[Tagged grammar]\label{def:taggedGrammar}
Let  $G=(V_N,\Sigma, P, S)$ be a grammar.
Define a language substitution $\rho: V \to \mathcal{P}\left( 
V \cup \Delta  \cup \left(\Sigma \cup \Delta \right)^2 \right)$ 
such that
\begin{equation*}
\rho(a) =  \left\{ a,\  a\odot,\  a],\  [a \right\} ,   a \in \Sigma; \qquad
\rho(X) = \left\{X \right\} \cup \Delta, \, X \in V_N.
\end{equation*}
Let $R$ be the regular language defined by 
$
R = \left(V_N \cup \{ [ \}\right) \cdot \Sigma \cdot \left( \left( V_N \cup \{ \odot \} \right) \cdot \Sigma \right)^* \cdot  \left( V_N \cup \{ ] \}\right)
$. 
We construct from $G$
 the \emph{tagged grammar} associated to $G$, denoted by
$\overline G =( V_N,\Sigma \cup \Delta, \overline P, S)$. 
For each rule $X \to R_X \in P$, $\overline G$ has the rule
$X \to \overline R_X$  where
$\overline R_X = \rho(R_X) \cap R$. \qed
\end{definition}
The idea underlying $\overline G$'s definition is to insert tags into the rules of $\overline P$, 
so 
that $G$'s structure becomes visible  in the  tagged words generated by $\overline G$.
Tagged grammars 
are  akin to the classical parenthesis grammars \cite{McNaughton67}, yet their representation of 
nested  structures is more parsimonious, since a single ``['' tag (analogously a ``]'')
can represent many open (resp. closed) parentheses.
Notice that $\sigma(L(\overline G)) \supseteq  L(G)$,
since tagged grammar rules exist, which replace a nonterminal with a tag. Such rules may generate words that, after deleting the tags, are not in $L(G)$.
To illustrate, going back to $G_1$ of Example~\ref{ex:opgram}, grammar $\overline G_1$ has the rules
$\{
S \to  \left(X \cup [ \right) b \left(X \cup \, ] \right)
$,
$X \to [a \left(\odot a\right)^* ]
\}
$
and generates the word $[b]$, while $\sigma([b])\notin L(G_1)$.
With the help of the  tagged grammar $\overline G$, we can compute all the tagged $k$-words that may occur in parsing any valid word for grammar $G$ (exemplified in \eqref{eqOPreduction}). Then,  we can check whether they are conflictual or not. In the latter case, grammar $G$ fulfills the next Definition \ref{defHOP}  of HOP($k$) grammar.  
\par
Returning to  Example~\ref{ex:opgram}, the tagged 3-words  $\varphi_3\left(\# L(\overline G_1) \# \right)$ coincide with the set $\Phi_3$  encoding the precedence relations. 
As observed, since $G_1$ is an OP grammar, $\Phi_3$ is 
nonconflictual, and $G_1$ is a HOP($3$) grammar as well. The formalization follows.

\begin{definition}[Higher-order Operator Precedence grammars]\label{defHOP} 
Let $k \ge 3$ be an odd integer. A  grammar  $G$, having  $\overline G$ as associated 
tagged grammar, is a {\em higher-order operator precedence grammar}  
of order $k$ (in short HOP($k$)) if 
\begin{equation}
\nexists u, v \in \varphi_k\left(\csharp  L(\overline G)  \csharp \right) 
\text{ such that } u \neq v \text{ and } \sigma(u) = \sigma(v)
\label{eqConflict}
\end{equation}
This means that the set of 
all tagged $k$-words  occurring in any sentence of $L(\overline G)$ is nonconflictual. 
\noindent
The union of the families $\ftaggable(k)$ for all values of $k$ is denoted by $\ftaggable$. 
The family of  grammars  $\taggable(k)$ having the same set $\Phi$  of tagged $k$-words is 
denoted by $\taggable(k, \Phi)$.  Identical notations denote the corresponding language families, when no confusion arises.
\qed
\end{definition} 
The decidability of Condition~(\ref{eqConflict}) for a given grammar and a 
fixed value of $k$ is obvious.
With an abuse of terminology, we also say that $\varphi_k(L(\overline G) )$ 
are the tagged $k$-words of grammar $G$. 

\begin{theorem}\label{teo:op}
The family  OP of \emph{operator precedence} languages coincides with the family
HOP(3), and is properly included within the HOP family.
\end{theorem}
\begin{proof}
	The proof formalizes the already stated fact that 
	OP relations are encoded by tagged 3-words. Let $G$ be an ECF grammar.
For  all letters $a,b \in \Sigma$ we show that the following relations hold: 
\begin{center}
$
\begin{array}{ccc}
a \,\dot=\, b  \iff 
a \odot b \in \varphi_3(L(\overline G) ),
&
a \lessdot b  \iff 
a [ b \in \varphi_3 (L(\overline G) ),
&
a \gtrdot b  \iff 
a ] b \in \varphi_3 (L(\overline G) ).
\end{array}
$
\end{center}
If  $a \gtrdot b$, from the definition of OP grammar \cite{Floyd1963}, it 
follows that there are a sentential word $u X v$ with $i_1(v)= b$,
and  an $X$-grammatical word $w$ such that $t_1(w)=a$ or $t_2(w)=aY$ with $Y \in V_N$.  
In both cases,  for the tagged grammar $\overline G$ (see Definition \ref{def:taggedGrammar}), 
either  the substitution 
 $\rho(a) = a]$ or $\rho(Y)= \ 
]$ 
causes the 3-word $a ] b$ to be in $\Phi_3$, the tagged $k$-words of grammar $G$. 
(Notice that a wrong  choice for 
the symbol returned by $\rho(Y)$, i.e. $[$ and $\odot$,  is neutralized by the 
intersection with language $R$ and  does not show up in the tagged grammar.) 
Conversely, it is obvious that $a ] b \in \Phi_3$ implies $a \gtrdot b$.
\par
We omit the similar case $a \lessdot b$, and examine the case $a \dot= b$,  
which happens if 
there exists a rule containing in the right part as factor $ab$ or $aYb$. The 
respective substitutions  $\rho(a) = a\odot$ and $\rho(Y) = \odot$ produce the 
3-word $a \odot b \in \Phi_3$. The converse is also immediate. 
It follows that, for every pair $a, b\in \Sigma$,  grammar $G$ violates the OP condition  if, 
and only if, the HOP condition  is false for $k=3$.

On the other hand, we show a $\taggable$ language that is not an
OP language. Let $L=\{a^n (baab)^n \mid n \geq 1\}$.
For any grammar of $L$,
by applying a pumping
lemma, it is clear that the relations $a\lessdot a$ and
$b \gtrdot b$ are unavoidable, i.e.,  the 3-words $a[a, b]b$ are necessarily
present in $\Phi_3$.
But it can be checked that, no matter how we choose the
other precedence relations, either there is a conflict or the language generated by the
grammar fails to be $L$; this can be exhaustively proved by examining  all
possible non-conflictual choices of $\Phi_3\subset \Sigma^{\square 3}$.
\par
On the other hand, it is possible to check that the grammar 
$G_2 : S \to a S b a a b \cup a b a a b$ is in HOP(7), and its tagged 7-words 
are

$
\Phi_7 = \left\{
\begin{array}{llllll}
\# \odot \#[a \odot b, &
a \odot a \odot b]\#, &
\#[a[a \odot b, &
a[a \odot b \odot a, &
b]b \odot a \odot a, &
a \odot b]b \odot a,\\
a \odot a \odot b]b, &
a \odot b]\# \odot \#, &
a[a[a[a, &
b \odot a \odot a \odot b, &
\#[a \odot b \odot a, &
\# \odot \# \odot \# \odot \#,\\
a[a[a \odot b, &
\# \odot \# \odot \#[a, &
b]\# \odot \# \odot \#, &
\# \odot \#[a[a, &
\#[a[a[a, &
a \odot b \odot a \odot a 
\end{array}
\right\}
.$

\end{proof}

\par
It is known  that OP grammars  are structurally unambiguous, and  that OP 
languages are CF deterministic and reverse-deterministic.  
Since such properties immediately follow from the bottom-up parser for OP 
grammars, which is easily extended to 
$\taggable(k)$ grammars without changing the essential operations, 
the same properties  hold for any value of $k$.

\begin{theorem}\label{th:structUnambiguityDET}
    Every HOP grammar is  
structurally unambiguous. The HOP language family is properly included within 
the  deterministic and 
reverse-deterministic CF languages.
\end{theorem}
\proof
We only need to prove the last statement. Let $L=\{a^n b a^n \mid n \geq 1\}$, 
which is deterministic and reverse deterministic. 
For any grammar $G$ of $L$, a straightforward application of the pumping lemma shows that, for any $k\geq 3$,  
$\varphi_k(L(\overline G))$
includes a 
word $a s_1 a s_2 \dots  a$  containing two distinct tags ``$[$'' and ``$]$'', 
therefore it also includes two conflictual $k$-words, 
because of the 
remark following  Lemma \ref{lemmaSLTlistLang}.
\qed

We show the significant connection between the $\taggable$ languages  and the $\maxtag$ languages, which motivates their appellation of max-languages.

\medskip
\noindent {\bf Max-grammars }
We prove by a fairly articulate construction that if $L$ is a max-language, 
i.e. 
$L = \ltaggable(\Phi_k)$, for some $\Phi_k \subseteq \Sigma^{\square k}$, 
then $L$ is generated by a grammar $G \in \taggable(k, \Phi_k)$. Moreover, $L$ 
is the largest language in $\taggable(k, \Phi_k)$.

Preliminarily, we define, for an arbitrary SLT language
 over a generic alphabet $\Upsilon$, a nondeterministic FA   which is symmetrical 
w.r.t. the scanning direction; this property contrasts with the standard deterministic sliding-window 
device, e.g., in~\cite{DBLP:journals/tcs/Caron00a}.
\begin{definition}[symmetrical FA]
	Let $F_k$ be a $k$-word set. 
	The $k$-\emph{symmetrical automaton} $A$ associated to $F_k$ is  obtained 
	by trimming the FA $A_0 = (\Upsilon, Q, \delta, I , T)$ where:
	\begin{itemize}
		\item $Q= \{ (\beta, \alpha)\in \Upsilon ^{k-1} \times \Upsilon ^{k-1} 
		\mid  
		\beta, \alpha \in f_{k-1}(F_k) \}$
		\item
		$(\beta, \alpha) \stackrel a \to (\beta', \alpha') \in \delta \text{ 
		if, and only if, }
		\beta'= t_{k-1}(\beta \, a )  \wedge \alpha= i_{k-1}(a \, \alpha')$
		\item
		$I= \{ (t_{k-1}(\csharp), \alpha)\in Q \}$,
		$T=\{ (\beta, i_{k-1}(\csharp)) \in Q \}$. \qed
	\end{itemize}
\end{definition}	
Intuitively, $\beta$ and  $\alpha$ represent the look-back and look-ahead 
	($k-1$)-words of  state
	$(\beta, \alpha)$.

\par
See Figure \ref{fig:simmaut}  for illustration. Two relevant properties of the 
\emph{symmetrical FA} $A$
are:
\begin{enumerate}
	\item $
	L(A)= \Loc(F_k)
	$, since, on each accepting path, the $k$-factors are by construction those of 
	$F_k$.
	
	\item The automaton $A$ is unambiguous. Consider a word $x = uyv$, and 
	assume by contradiction that there are two accepting paths in $A$, with state sequences 
	$\pi_u \pi_y \pi_v$, $\pi_u \pi'_y \pi_v$ and the same label $x$ ($u$ and 
	$v$ could be $\varepsilon$). But, by construction of $A$, if $\pi_y = q_1 
	q_2 \ldots q_t$ and  $\pi'_y = q'_1 q'_2 \ldots q'_t$, then $q_1 = 
	(t_{k-1}(u), i_{k-1}(y))$ and also $q'_1 = (t_{k-1}(u), i_{k-1}(y))$. This 
	holds for every subsequent step, hence $\pi'_y = \pi_y$, so the two paths 
	must be identical.
\end{enumerate}

Given  $\Phi \subseteq \Sigma^{\square k}$, we construct a grammar, 
denoted by  $\overline{G}_\Phi$, that generates the max-language $\ltagged(\Phi)$. The 
construction is also illustrated in Example~\ref{ex:grammar}. 

\begin{definition}[max-grammar 
construction]\label{defMaxGrammarConstruction}\label{def:grammgraph}

	Let $\Phi \subseteq \Sigma^{\square k}$, 
and let $A_\Phi=(\Sigma\cup \Delta, Q, \delta, I , T)$ be the symmetrical 
automaton recognizing $\Loc(\Phi)$.
	The grammar $\overline{G}_\Phi = (\Sigma \cup \Delta, V_N, P, S)$  
	called \emph{tagged max-grammar}, 
	is obtained by reducing (in the sense of trimming the useless parts) the grammar constructed as follows.
	\begin{itemize}
		\item 
		$V_N$ is a subset of $Q \times Q$ such that 
		$( q_1, q_2 ) \in V_N$ if $q_1 = (\beta_1, [ \gamma_1)$
		and   
		$q_2= (\gamma_2 ], \alpha_2)$, for some $\beta_1, \gamma_1, \gamma_2, 
		\alpha_2$. Thus a nonterminal $X$ is also identified by $( 
		\beta_1, [ \gamma_1, \gamma_2 ], \alpha_2  )$.
		
		\item The axiom set is $S = I \times T$.
		
		\item 
		Each rule  $X \to R_X \in P$ is such that the right part is defined by 
        $M_X = \left(V \cup \Delta, Q_X, \delta_X, \{p_{I}\}, \{p_{T}\} 
        \right)$,
an FA		where $Q_X \subseteq Q$,    next specified.
		
		Let $X = ( \beta_X, \alpha_X, \beta'_X, \alpha'_X )$. Then:
		$p_I = (\beta_X, \alpha_X)$, $p_T = (\beta'_X, \alpha'_X)$,
	    \item
		The graph of the transition relation is $\delta_X = (\delta \cup \delta') - \delta''$, where \\
		$
		\delta' = \left\{ 
		( \beta_1, \alpha_3 ) 
		\stackrel
		{( \beta_1, \alpha_1, \beta_2, \alpha_2 )}
		{\longrightarrow}
		( \beta_3, \alpha_2 )
		\ \vrule
		\begin{array}{l}
		( \beta_1, \alpha_1, \beta_2, \alpha_2 ) \in V_N,\\
		( \beta_1, \alpha_3 ) 
		\stackrel
		{\odot}
		{\longrightarrow}
		( \beta_3, \alpha_2 )
		\in \delta
		\ \lor \\
		( \beta_1, \alpha_3 )  = p_I \land
		( \beta_1, \alpha_3 ) 
		\stackrel
		{[}
		{\longrightarrow}
		( \beta_3, \alpha_2 )
		\in \delta
		\ \lor \\
		( \beta_3, \alpha_2 )  = p_T \land
		( \beta_1, \alpha_3 ) 
		\stackrel
		{]}
		{\longrightarrow}
		( \beta_3, \alpha_2 )
		\in \delta
		
		\end{array}
		\right\}
		$\\
		$
		\delta'' =
		\begin{array}{l}
		\left\{
		q' \stackrel{[}{\longrightarrow} q'' \in \delta \ \vrule\  q' \ne p_I
		\right\}
		\cup
		\left\{
		q' \stackrel{]}{\longrightarrow} q'' \in \delta  \ \vrule\  q'' \ne p_T
		\right\}
		\cup\\
		\left\{
		q' \stackrel{x}{\longrightarrow} p_I \in \delta \right\}
		\cup
		\left\{
		p_T \stackrel{x}{\longrightarrow} q' \in \delta \right\}.
		\end{array}
		$
	\end{itemize}
	Intuitively, $\delta'$ adds transitions with nonterminal labels between any 
	two states already 
	linked by a tag-labeled transition, which are  ``compatible'' with the 
	nonterminal name (i.e. with the same look-back and look-ahead). The transitions $\delta''$ to be deleted are: those labeled by
	tags ``[`` or ``]'' that are not initial or final, and those reentering the initial 
	or final states. 

	Define the \emph{max-grammar} as 
	$G_\Phi = \left(V_N, \Sigma, \left\{ X \to \sigma(R_X) \mid X \to R_X \in P \right\}, 
	S \right)$.
	
	The \emph{grammar graph} $\Gamma(\overline{G}_\Phi)$ of $\overline{G}_\Phi$ 
	is a graph containing all the arcs and states of the symmetrical automaton 
	$A_\Phi$ associated to $\Phi$, together with all the arcs labelled by nonterminals,
    defined by the above construction of $\delta'$.
\qed
\end{definition}

The grammar graph   synthetically represents all the rules of the 
max-grammar $\overline{G}_\Phi$ and will be used in the proof of the forthcoming lemma.
Each rule right part is a subgraph starting 
with a label ``[``,  ending with a label ``]'', and containing only terminals and  
$\odot$ tags; the rule left part is denoted by the pair of initial and final 
states of the subgraph.

\begin{figure}
	\begin{center}
		\scalebox{1.1}{
			\begin{tikzpicture}[every edge/.style={draw,solid}, 
			node distance=2.8cm, 
			auto, 
			every state/.style={draw=black!100,scale=0.65}, 
			>=latex]
			
			\node[initial by arrow, initial text=,state] (q1) 
			{{$ (\odot \#, [a )$}};		
			\node[state] (q2)  [below of=q1, xshift=0cm]  {{$( 
					\#[, 
					a[  
					)$}};
			\node[state] (q4)  [below of=q2, 
			xshift=0cm]  {{ $( [a, [a )$}};
			
			\node[state] (q3)  [right of=q1, xshift=0cm]  {{$( 
					\#[, a 
					\odot )$}};
			\node[state] (q5)  [below of=q3, xshift=0cm]  {{$( a[, a\odot 
					)$}};
			\node[state] (q6)  [below of=q5, xshift=0cm]  {{$( 
					a[, a[ 
					)$}};
			
			\node[state] (q7)  [right of=q3, xshift=0cm]  {{$( [a, \odot 
					c )$}};
			\node[state] (q8)  [below of=q7, xshift=0cm]  {{$( 
					a\odot, c\odot )$}};
			
			\node[state] (q12) [right of=q7, xshift=0cm]  {{$( b\odot, 
					b\odot )$}};
			
			\node[state] (q11) [right of=q12, xshift=0cm] {{$( \odot b, 
					\odot b )$}};
			\node[state] (q10) [below of=q11, xshift=0cm] {{$( c \odot, b 
					\odot )$}};
			\node[state] (q9)  [below of=q10, xshift=0cm] {{$( \odot c, 
					\odot b )$}};
			
			\node[state] (q13) [right of=q11, xshift=0cm] {{$( b\odot, b] 
					)$}};
			\node[state] (q14) [below of=q13, xshift=0cm] {{$( \odot b, 
					]c )$}};
			\node[state] (q15) [below of=q14, xshift=0cm ]
			{{$ ( b], c\odot ) $}};
			
			\node[state] (q16) [right of=q15, xshift=0cm] {{$( ]c, 
					\odot b )$}};
			
			\node[state] (q17) [right of=q16, xshift=0cm] {{$( c \odot, 
					b] )$}};
			\node[state] (q18) [above of=q17, xshift=0cm] {{$( 
					\odot 
					b,]\# )$}};
			\node[state] (q19) [above of=q18, xshift=0cm, accepting] 
			{{$ ( b ], \# \odot )$}};
			\path[->]
			
			(q1) edge  node {$[$} (q2)
			(q1) edge  node {$[$} (q3)
			(q2) edge  node {$a$} (q4)
			(q3) edge  node {$a$} (q7)
			(q4) edge [bend right, above] node {$[$} (q6)
			(q6) edge [bend right, above] node {$a$} (q4)
			(q4) edge node {$[$} (q5)
			(q5) edge node {$a$} (q7)
			(q7) edge [left] node {$\odot$} (q8)
			(q8) edge [bend right, above] node {$c$} (q9)
			(q9) edge node {$\odot$} (q10)
			(q10) edge node {$b$} (q11)
			(q9) edge [bend right, below] node {$\odot$} (q17)
			(q11) edge [bend right, above] node {$\odot$} (q12)
			(q12) edge [bend right, above] node {$b$} (q11)
			(q11) edge node {$\odot$} (q13)
			(q13) edge node {$b$} (q14)
			(q13) edge [bend left] node {$b$} (q18)
			(q14) edge node {$]$} (q15)
			(q15) edge  node {$c$} (q16)
			(q17) edge [bend right, above] node {$b$} (q14)
			(q16) edge  node {$\odot$} (q17)
			(q17) edge  node {$b$} (q18)
			(q18) edge  node {$]$} (q19)
			;
 			\end{tikzpicture}}
            \\
            (i)
		\scalebox{1.1}{
			\begin{tikzpicture}[every edge/.style={draw,solid}, 
			node distance=2.8cm, 
			auto, 
			every state/.style={draw=black!100,scale=0.65}, 
			>=latex]
			
			\node[initial by arrow, initial text=,state] (q1) 
			{{$ (\odot \#, [a )$}};		
			\node[state] (q3)  [right of=q1, xshift=0cm]  {{$( 
					\#[, a 
					\odot )$}};
			
			\node[state] (q7)  [right of=q3, xshift=0cm]  {{$( [a, \odot 
					c )$}};
			\node[state] (q8)  [below of=q7, xshift=0cm]  {{$( 
					a\odot, c\odot )$}};
			
			\node[state] (q12) [right of=q7, xshift=0cm]  {{$( b\odot, 
					b\odot )$}};
			
			\node[state] (q11) [right of=q12, xshift=0cm] {{$( \odot b, 
					\odot b )$}};
			\node[state] (q10) [below of=q11, xshift=0cm] {{$( c \odot, b 
					\odot )$}};
			\node[state] (q9)  [below of=q10, xshift=0cm] {{$( \odot c, 
					\odot b )$}};
			
			\node[state] (q13) [right of=q11, xshift=0cm] {{$( b\odot, b] 
					)$}};
			\node[state] (q17) [right of=q16, xshift=0cm] {{$( c \odot, 
					b] )$}};
			\node[state] (q18) [above of=q17, xshift=0cm] {{$( 
					\odot 
					b,]\# )$}};
			\node[state] (q19) [above of=q18, xshift=0cm, accepting] 
			{{$ ( b ], \# \odot )$}};
			\path[->]
			
			(q1) edge  node {$[$} (q3)
			(q3) edge  node {$a$} (q7)
			(q7) edge [left] node {$\odot$} (q8)
			(q8) edge [bend right, above] node {$c$} (q9)
			(q9) edge node {$\odot$} (q10)
			(q10) edge node {$b$} (q11)
			(q9) edge [bend right, below] node {$\odot$} (q17)
			(q11) edge [bend right, above] node {$\odot$} (q12)
			(q12) edge [bend right, above] node {$b$} (q11)
			(q11) edge node {$\odot$} (q13)
			(q13) edge [bend left] node {$b$} (q18)
			(q17) edge  node {$b$} (q18)
			(q18) edge  node {$]$} (q19)
			(q7) edge [bend left] node {$( [a, [a, b], c \odot )$} (q8)
			;
			\end{tikzpicture}}
            \\
            (ii)
            \\
		\scalebox{1.1}{
			\begin{tikzpicture}[every edge/.style={draw,solid}, 
			node distance=2.8cm, 
			auto, 
			every state/.style={draw=black!100,scale=0.65}, 
			>=latex]
			
			\node[initial by arrow, initial text=,state] (q4)  [below of=q2, 
			xshift=0cm]  {{$ ( [a, [a ) $}};
			
			\node[state] (q5)  [below of=q3, xshift=0cm]  {{$( a[, a\odot 
					)$}};
			
			\node[state] (q7)  [right of=q3, xshift=0cm]  {{$( [a, \odot 
					c )$}};
			\node[state] (q8)  [below of=q7, xshift=0cm]  {{$( 
					a\odot, c\odot )$}};
			
			\node[state] (q12) [right of=q7, xshift=0cm]  {{$( b\odot, 
					b\odot )$}};
			
			\node[state] (q11) [right of=q12, xshift=0cm] {{$( \odot b, 
					\odot b )$}};
			\node[state] (q10) [below of=q11, xshift=0cm] {{$( c \odot, b 
					\odot )$}};
			\node[state] (q9)  [below of=q10, xshift=0cm] {{$( \odot c, 
					\odot b )$}};
			
			\node[state] (q13) [right of=q11, xshift=0cm] {{$( b\odot, b] 
					)$}};
			\node[state] (q14) [below of=q13, xshift=0cm] {{$( \odot b, 
					]c )$}};
			\node[state] (q15) [below of=q14, xshift=0cm, accepting]
			{{$ ( b], c\odot )$}};
			
			\node[state] (q17) [right of=q16, xshift=0cm] {{$( c \odot, 
					b] )$}};
			\path[->]
			
			(q4) edge node {$[$} (q5)
			(q5) edge node {$a$} (q7)
			(q7) edge [left] node {$\odot$} (q8)
			(q7) edge [bend left] node {$( [a, [a, b], c \odot )$} (q8)
			(q8) edge [bend right, above] node {$c$} (q9)
			(q9) edge node {$\odot$} (q10)
			(q10) edge node {$b$} (q11)
			(q9) edge [bend right, below] node {$\odot$} (q17)
			(q11) edge [bend right, above] node {$\odot$} (q12)
			(q12) edge [bend right, above] node {$b$} (q11)
			(q11) edge node {$\odot$} (q13)
			(q13) edge node {$b$} (q14)
			(q14) edge node {$]$} (q15)
			(q17) edge [bend right, above] node {$b$} (q14)
			;
			\end{tikzpicture}}
            \\
            (iii)
		\caption{
			(i) Symmetrical FA $A_\Phi$ of Example~\ref{ex:FirstExample}.
			(ii) Automaton of the rule 
            $X 
            \to 
            [ a (\odot \cup  Y) c \odot (b \odot)^* b ]$ 
            of 
			Example~\ref{ex:grammar}.
			(iii) Automaton of the rule  
            $Y 
            \to 
            [ a (\odot \cup  Y) c \odot (b \odot)^* b ]$ 
            of
			Example~\ref{ex:grammar}.
		}\label{fig:rule}\label{fig:simmaut}								
		
	\end{center}
\end{figure}
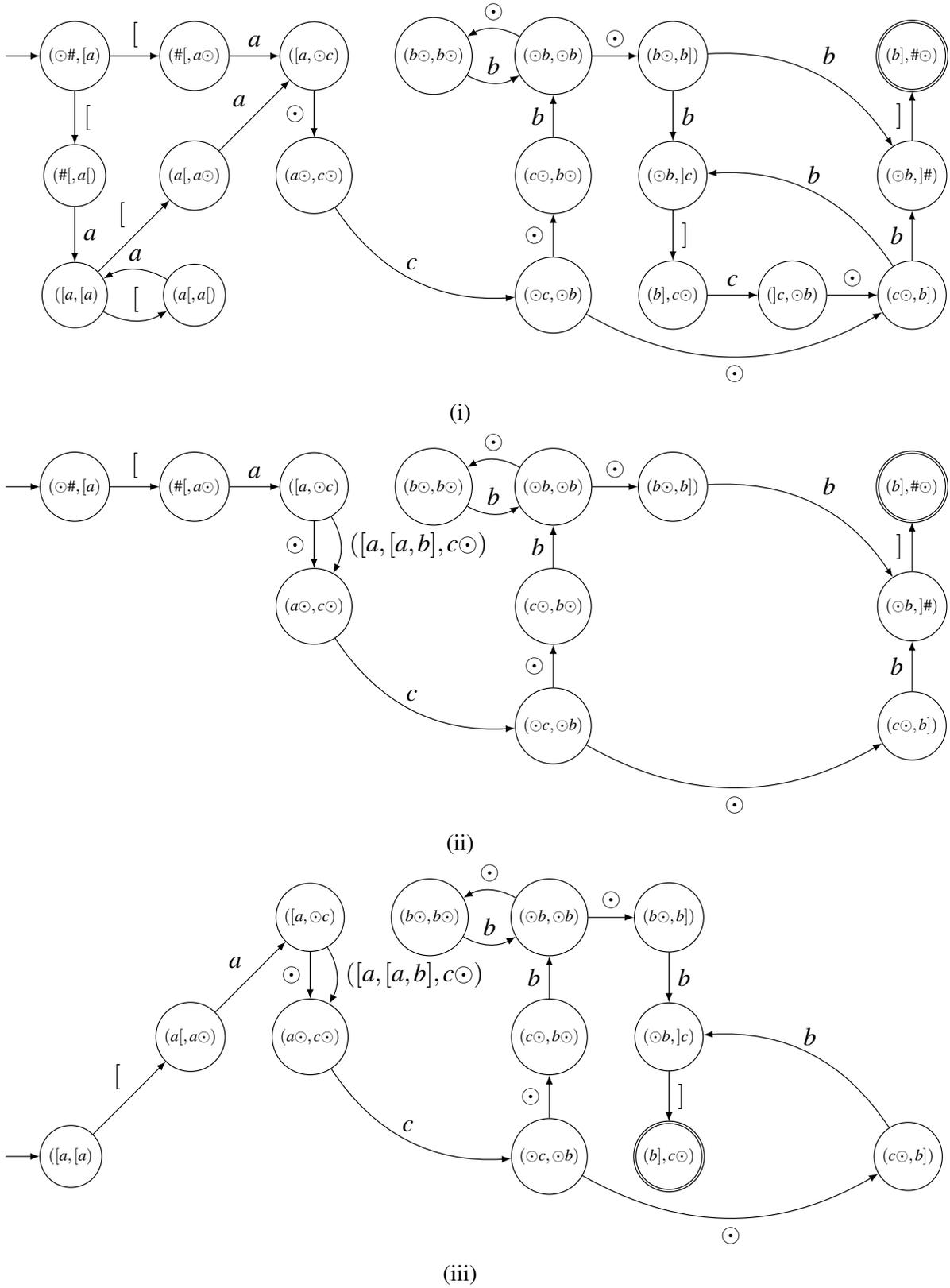

\begin{example}\label{ex:grammar}
	We show the construction of the max-grammar  for the 
	tagged 3-word set $
\Phi= \{ 
\#\odot\#,$
$\# [ a,\ 
b ] \#,\ $ 
$b ] c,\ 
c \odot b,\ 
b \odot b,\ 
a \odot c,\ 
a [ a
\}
$ of Example~\ref{ex:FirstExample}.
Its symmetrical automaton $A_\Phi$ is reported in Figure~\ref{fig:simmaut}~(i).
	The nonterminals are included in the set
	$\{ 
	( \odot \#, [ a), 
	( [ a, [ a )
	\}
	\times
	\{ 
	( b ], c \odot ), 
	( b ], \# \odot )
	\}
	$, but 
	$( \odot \#, [ a, b], c \odot)$ and
	$( [ a, [ a, b ], \#  \odot )$
	are unreachable, because they are neither axioms nor they are  transition 
	labels in the grammar graph.
    Thus only two nonterminals are left: 
	the axiom $X = ( \odot \#, [a, b], \# \odot )$ 
    and $Y = ( [a, [a, b], c \odot )$ which occurs on the transition
	from $( [a, \odot c )$ to $( a \odot, c \odot )$.
    The two rules of the resulting grammar are 
            $X           \to 
            [ a (\odot \cup  Y) c \odot (b \odot)^* b ]$
            and
            $Y  
            \to 
            [ a (\odot \cup  Y) c \odot (b \odot)^* b ]$ 
    and their automata are show in Figure~\ref{fig:rule}~(ii) and (iii), respectively.
\end{example}

\par
By construction, the rules of any tagged max-grammar  $\overline{G}_\Phi$ have some 
properties worth noting:
\begin{enumerate}
	\item 
	For each rule $X \to M_X$, 
	$R_X \subseteq 
	(V_N \cup \{[\}) \cdot \Sigma \cdot
	\left(
	\left(
	V_N \cup \{ \odot \}
	\right)
	\cdot \Sigma
	\right)^*
	\cdot
	(V_N \cup \{]\}).
	$
	This fact implies that $\overline{G}_\Phi$ and $G_\Phi$ are in operator 
	form.  
	
	\item For each rule $X \to M_X$ in $P$, $M_X$ is an unambiguous FA.
\end{enumerate}

\noindent
We prove that the languages defined by  max-grammars  and  by reductions of 
Definition~\ref{defSLRlanguage} coincide.

\begin{lemma}\label{thReductionsEquivGrammar}
	Let $\Phi \subseteq \Sigma^{\square k}$, and let $\overline{G}_\Phi$ and 
	$G_\Phi$ be  the max-grammars of Definition~\ref{defMaxGrammarConstruction}. 
	Then 
	$L(\overline{G}_\Phi)= \ltagged(\Phi)$ and 
	$L(G_\Phi) = \ltaggable(\Phi)$.
\end{lemma}

\begin{proof}
	It suffices to consider $\overline{G}_\Phi$, since $G_\Phi$ 
	has the same structure.
	We need also  the symmetrical FA $A_\Phi $, and the grammar graph 
	$\Gamma(\overline{G}_\Phi)$.
	Notice that  $A_\Phi $ and $\Gamma(\overline{G}_\Phi)$ have the same set of 
	states, and that $A_\Phi $ is a sub-graph of $\Gamma(\overline{G}_\Phi)$, 
	which only differs by the absence of nonterminally-labeled arcs.
	\par
	We say that two  words $w$ and $w'$   are equivalent on a sequence of states $\pi = 
	q_1, q_2, \ldots, q_n $ (or \emph{path equivalent}), written  $w \equiv_\pi 
	w'$, 
	iff in $\Gamma(\overline{G}_\Phi)$ there exist two paths, both with the state sequence 
	$\pi$, such that $w$ and $w'$ are their labels.
	
	\par
	We start from a string $w^{(0)} \in \Loc(\Phi)$; we will show that, for 
	some 
	$m > 0$ and for some axiom $W \in S$:\\
	$
	w^{(0)} \leadsto_{\Phi} 
	w^{(1)} \leadsto_{\Phi} 
	\ldots \leadsto_{\Phi} 
	w^{(m)} = \csharp \odot \csharp
	\ \text{ iff } \ 
	\tilde{w}^{(0)} \Longleftarrow_{\overline{G}_\Phi} 
	\tilde{w}^{(1)} \Longleftarrow_{\overline{G}_\Phi} 
	\ldots \Longleftarrow_{\overline{G}_\Phi} 
	\tilde{w}^{(m)} = W,
	$\\
	where 
	$ \tilde{w}^{(0)}  = w^{(0)}$,
	and $\forall i$, $\exists \pi_{i}$ : 
	$ \tilde{w}^{(i)}  \equiv_{\pi_i} w^{(i)}$.
	
	\par
	We prove the theorem by induction on the reduction steps.
	
	\noindent {\bf Base case:} 
	Consider $\tilde{w}^{(0)}$: it is by definition  
	$ \tilde{w}^{(0)}  = w^{(0)}$, hence
	$\exists \pi :  \tilde{w}^{(0)}  \equiv_\pi w^{(0)}$.
	
	\noindent {\bf Induction case:}
	First, we prove that $w^{(t)} \leadsto_{\Phi} w^{(t+1)}$ implies
	$\tilde{w}^{(t)} \Longleftarrow_{\overline{G}_\Phi} \tilde{w}^{(t+1)}$
	with $\tilde{w}^{(t+1)} \equiv_{\pi_{t+1}} w^{(t+1)}$.
	To perform the reduction, we need a handle, let it be called $x$, such that 
	$w^{(t)} = u x v \leadsto w^{(t+1)} = u s v$, $s \in \Delta$. 
	By induction hypothesis, we know that 
	$ \tilde{w}^{(t)} \equiv_{\pi_t} w^{(t)} =  u x v$, 
	therefore 
	$\tilde{w}^{(t)} = \tilde{u} \tilde{x} \tilde{v}$ with
	$\tilde{u} \equiv_{\pi'_t} u$, 
	$\tilde{x} \equiv_{\pi''_t} x$, 
	and
	$\tilde{v} \equiv_{\pi'''_t} v$, 
	with $\pi_t = \pi'_t  \pi''_t  \pi'''_t$. 
	The equivalence $\tilde{x} \equiv_{\pi''_t} x$, with $x$ handle, implies 
	that there is a right part 
	of a rule $X \to M_X \in P$, such that
	$\tilde{x} \in R_X$. 
	Hence,
	$\tilde{w}^{(t)} \Longleftarrow_{\overline{G}_\Phi} \tilde{w}^{(t+1)} = 
	\tilde{u} X \tilde{v}$ and
	$X = ( t_{k-1}(\csharp u), i_{k-1}(x v \csharp), t_{k-1}(\csharp u 
	x), i_{k-1}(v \csharp) )$.
	The reduction relation implies that in $A_\Phi $ (and therefore also in 
	$\Gamma(\overline{G}_\Phi)$) 
	there is a path with states $\pi_{t+1}$ and labels $w^{(t+1)}$: 
	call $\pi'_{t+1}$ the states of its prefix with label $u$, and 
	$\pi''_{t+1}$ those of its suffix with label $v$.
	Let us call $q_u$ the last state of $\pi'_{t+1}$ and $q_v$ the first state 
	of $\pi''_{t+1}$.
	By construction of $\overline{G}_\Phi$, in $\Gamma(\overline{G}_\Phi)$ 
	there is a transition 
	$q_u \stackrel{X}{\longrightarrow} q_v$,
	while in $A_\Phi $ there is 
	$q_u \stackrel{s}{\longrightarrow} q_v$.
	From this it follows $\tilde{w}^{(t+1)} \equiv_{\pi'_{t+1} \pi''_{t+1}} 
	w^{(t+1)}$. 
	\par
	We now prove that
	$\tilde{w}^{(t)} \Longleftarrow_{\overline{G}_\Phi} \tilde{w}^{(t+1)}$ 
	implies
	$w^{(t)} \leadsto_{\Phi} w^{(t+1)}$, 
	with $w^{(t+1)} \equiv_{\pi_{t+1}} \tilde{w}^{(t+1)}$.
	By definition of derivation, it is 
	$\tilde{w}^{(t)} = \tilde{u} \tilde{x} \tilde{v} 
	\Longleftarrow_{\overline{G}_\Phi} \tilde{w}^{(t+1)} = \tilde{u} X 
	\tilde{v}$ for some $X \in V_N$.
	By induction hypothesis, we know that 
	$\tilde{u} \tilde{x} \tilde{v} = \tilde{w}^{(t)} \equiv_{\pi_t} w^{(t)}$, 
	hence 
	$w^{(t)} = u x v$ with
	$\tilde{u} \equiv_{\pi'_t} u$, 
	$\tilde{x} \equiv_{\pi''_t} x$, 
	and
	$\tilde{v} \equiv_{\pi'''_t} v$, 
	with $\pi_t = \pi'_t  \pi''_t  \pi'''_t$. 
	From this it follows that
	$X = ( t_{k-1}(\csharp u), i_{k-1}(x v \csharp), t_{k-1}(\csharp u 
	x), i_{k-1}(v \csharp) )$, and 
	that $x$ must be an handle. Therefore,  
	$w^{(t)} = u x v \leadsto w^{(t+1)} = u s v$, $s \in \Delta$, and 
	in $A_\Phi$ (and in $\Gamma(\overline{G}_\Phi)$) 
	there is a path with states $\pi_{t+1}$ and labels $w^{(t+1)}$: 
	call $\pi'_{t+1}$ the states of its prefix with label $u$, and 
	$\pi''_{t+1}$ those of its suffix with label $v$.
	Let us call $q_u$ the last state of $\pi'_{t+1}$ and $q_v$ the first state 
	of $\pi''_{t+1}$.
	By construction of $\overline{G}_\Phi$, in $\Gamma(\overline{G}_\Phi)$ 
	there is a transition $(q_u, X, q_v)$, 
	while in $A_\Phi $ there is $(q_u, s, q_v)$. 
	Hence $\tilde{w}^{(t+1)} \equiv_{\pi'_{t+1} \pi''_{t+1}} w^{(t+1)}$.
\end{proof}

\begin{theorem}\label{th:InclusionMaxlang}
 Let $G$  be any  grammar in the family $\ftaggable(k, \Phi)$ and $\overline G$ 
 its tagged version. Let $\ltaggable(\Phi)= L(G_\Phi)$  (respectively 
 $\ltagged(\Phi)=L(\overline G_\Phi)$) 
 be the max-languages.
 The following inclusions hold:
\begin{center}
$
L(\overline G) \subseteq \ltagged(\Phi),\quad L(G) \subseteq  
\ltaggable(\Phi).
$
\end{center}
\end{theorem}
\proof (Hint)
Let $G=(V_N,\Sigma, P,  S)$, $\overline G=(V_N,\Sigma\cup \Delta, \overline P,  S)$, $G_\Phi=(V'_N, \Sigma, P', S' )$ and $\overline G_\Phi=(V'_N, \Sigma \cup \Delta, \overline P', S')$. 
We prove that if, for $X \in S$,  $X \stackrel + \Longrightarrow_{\overline{G}} w $ then, for some $Y' \in S'$, 
$Y' \stackrel + \Longrightarrow_{\overline{G}_\Phi} w$; we may assume both derivations are leftmost.
If
$\ X  \stackrel + \Longrightarrow_{\overline{G}} u X_1 v 
\Longrightarrow_{\overline{G}}  u w_1 v = w $
then $w_1$ is the leftmost handle in $w$, and by definition of maxgrammar, 
there exists a derivation
$u Z' v \Longrightarrow_{\overline{G}} u w_1 v $
where $Z'$ is the 4-tuple 
$( t_k(\csharp u), i_k(w_1 v \csharp), t_k(\csharp u w_1),$ $i_k(v \csharp ) )$.
\par
Then, after the reduction or the derivation step, the position of the leftmost
handle in $u X_1 v$ and in $u Z' v$ coincide, and we omit the simple inductive 
arguments that completes the proof.
\par
Clearly, the two derivations of $\overline{G}$ and of $\overline{G}_\Phi$ have the same length and create isomorphic trees, which only differ in the nonterminal names.
By applying the projection $\sigma$ to both derivations, the inclusion $L(G) \subseteq  \ltaggable(\Phi)$ follows.
\qed

\par
Thus,  for each set of  tagged $k$-words $\Phi$, the
max-language $L(G_\Phi)$ includes all languages in $\ftaggable(k, 
\Phi)$, actually also any language in  $\ftaggable(k, 
\Phi')$, where $\Phi'\subseteq \Phi$.

\par
To prove the Boolean closure of $\ftaggable(k, 
\Phi)$, we need the following lemma (the tedious proof is omitted)
which extends Theorem 5 of Knuth \cite{KnuthLR67} from CF to ECF grammars.
\begin{lemma}\label{lm:ComplementECF}
	Let  $G_{(\,),1}$ and 
	$G_{(\,),2}$ be  ECF parenthesis 
	grammars. Then there exists an ECF parenthesis grammar $G_{(\,)}$ such 
	that  
	$L(G_{(\,)}) = L(G_{(\,),1}) - L(G_{(\,),2})$. 
\end{lemma}
\begin{theorem}\label{th:BoolClosure}
For every $k$ and $\Phi \subset \Sigma^{\square k}$,  the language family  $\ftaggable(k, \Phi)$ is closed under union, intersection and under relative complement, i.e., $L_1 - L_2 \in \ftaggable(k, \Phi)$ if $L_1 , L_2 \in \ftaggable(k, \Phi)$.
\end{theorem}
\begin{proof} Let $L_i= L(G_i)$ where $G_i=(V_{N_i},\Sigma, P_i, S_i)$, for 
$i=1,2$. 
 We assume that the nonterminal names of the two grammars are disjoint.
\par
\textbf{Union}.  The grammar $G=(V_{N_1}\cup V_{N_2},\Sigma, P_1 \cup P_2, S_1 \cup S_2)$  generates $L(G_1) \cup  L(G_2)$ and  is  in   $\ftaggable(k, \Phi)$, since its set of tagged $k$ grams is $\Phi$.
\par
\textbf{Complement}. Let $G_{(\,),i}$ be the parenthesis grammar of $G_i$, 
$i=1, 2$, and by Lemma \ref{lm:ComplementECF} let $G_{(\,)}=(V_N, \Sigma, P, 
S)$ be the parenthesis grammar such that  $L(G_{(\,)}) = L(G_{(\,),1}) - 
L(G_{(\,),2})$.
Since $G_1$ and $G_2$ are structurally unambiguous by Theorem \ref{th:structUnambiguityDET},  there exists a bijection between the sentences of $L_i$ and $L(G_{(\,),i})$, $i=1, 2$.   
\par
Define the grammar $G=(V_N, \Sigma, P, S)$ obtained from $G_{(\,)}$ by erasing the parentheses from each rule right part. It is obvious that   $L(G)=L_1 - L_2$ since, if  $x, y$ are sentences of  $G_{(\,)}$ and $\sigma(x)=\sigma(y)$, then $x=y$.
It remains to prove that $G$ has the $\taggable(k)$ property. Each sentence of $L(G)$ corresponds to one, and only one, sentence of $G_{(\,)}$.
Since $L(G) \subseteq L_1$, the tagged $k$-words of grammar $G$ are  a subset of the tagged $k$-words $\Phi$ of grammar $G_1$, which by hypothesis are not conflictual.
The closure under intersection is obvious.
\end{proof}
Combining  Theorem \ref{th:InclusionMaxlang} and Theorem \ref{th:BoolClosure}, we have:
\begin{corollary}\label{Cor:BoolAlgebra}
For every $k$ and $\Phi \subset \Sigma^{\square k}$,   the language family  
$\ftaggable(k, \Phi)$ is a Boolean algebra having as top element the 
max-language $\ltaggable(\Phi)$.
\end{corollary}

Our last result reaffirms for the HOP languages a useful  property of OP languages.  
\begin{theorem}\label{th:reg-intersection}
    For every $k$ and $\Phi \in \Sigma^{\square k}$, the language family HOP$(k, \Phi)$
    is closed under intersection with regular languages.
\end{theorem}
\begin{proof} (Hint)
Let us consider a grammar $G_0 \in$ HOP$(k, \Phi)$ and a regular language $R_0$.
 We first add tags to $R_0$ through the language substitution 
 $\eta : \Sigma^2 \to \mathcal{P}(\Sigma \cdot \Delta \cdot  \Sigma)$, such that 
 $\eta(a b) = \{a\} \Delta \{b\}$. 
Consider the regular language
$R_1 = \eta(R_0)$ and an FA 
 $M_1 = (\Sigma \cup \Delta, Q_R, \delta_R, I_R, T_R)$ that recognizes $R_1$.
Let $A_\Phi$ be the symmetrical automaton recognizing $\Loc(\Phi)$. We  apply 
 the classic ``product'' construction
 for the language intersection of the two FA $A_\Phi$ and $M_1$;  let the product machine be
 $(\Sigma\cup \Delta, Q, \delta, I , T)$.
Note that a state of $Q$ consists of three components
$( \beta_1, \alpha_1, q_1)$: 
the look-back $\beta_1$ and  look ahead $\alpha_1$, where  $\beta_1, \alpha_1 \in (\Sigma \cup \Delta)^{k-1}$
come from the states of $A_\Phi$,  while the
 state $q_1$ comes from  $Q_R$.

 By extending the construction presented in Definition~\ref{defMaxGrammarConstruction}, we proceed now  to define the grammar
 $G_1 = (\Sigma \cup \Delta, V_N, P, S)$
 for $\ltagged(\Phi) \cap R_1$ as follows.
		
\noindent -- $V_N$ is a subset of $Q \times Q$ such that 
		$( \beta_1, [ \gamma_1,  q_1, \gamma_2 ], \alpha_2, q_2 ) \in V_N$, for some $\beta_1, \gamma_1, \gamma_2, 
		\alpha_2$, $q_1$, $q_2$.
		
\noindent --	$S \subseteq V_N$ and $X \in S$ if, and only if, 
		$X = ( \gamma_1 \odot \#, \alpha_1, q_1, \beta_1, \# \odot \gamma_2, q_2 
		)$, for some $\gamma_1, \gamma_2, \alpha_1, \beta_1$, $q_1 \in I_R$, $q_2 \in F_R$.
		
\noindent --		Each rule  $X \to M_X \in P$ is such that the right part is an FA 
        $M_X = \left(\Sigma\cup \Delta, Q_X, \delta_X, \{p_{I}\} , \{p_{T}\} \right)$ 
		where
		$Q_X \subseteq Q$.  
		(For each $X$  there exists only one $M_X$.)
		Let $X = ( \beta_X, \alpha_X, q_X, \beta'_X, \alpha'_X, q'_X )$. Then:
		$p_I = (\beta_X, \alpha_X, q_X)$, $p_T = (\beta'_X, \alpha'_X, q'_X)$,
		$\delta_X = (\delta \cup \delta') - \delta''$, \\
		$
		\delta' = \left\{ 
		( \beta_1, \alpha_3, q_1 ) 
		\stackrel
		{( \beta_1, \alpha_1, q_1, \beta_2, \alpha_2, q_2 )}
		{\longrightarrow}
		( \beta_3, \alpha_2, q_2 )
		\ \vrule
		\begin{array}{l}
		( \beta_1, \alpha_1, q_1, \beta_2, \alpha_2, q_2 ) \in V_N,\\
		( \beta_1, \alpha_3, q_1 ) 
		\stackrel
		{\odot}
		{\longrightarrow}
		( \beta_3, \alpha_2, q_2 )
		\in \delta
		\  \lor \\
		( \beta_1, \alpha_3, q_1 )  = p_I \land
		( \beta_1, \alpha_3, q_1 ) 
		\stackrel
		{[}
		{\longrightarrow}
		( \beta_3, \alpha_2, q_2 )
		\in \delta
		\  \lor \\
		( \beta_3, \alpha_2, q_2 )  = p_T \land
		( \beta_1, \alpha_3, q_1 ) 
		\stackrel
		{]}
		{\longrightarrow}
		( \beta_3, \alpha_2, q_2 )
		\in \delta
		\end{array}
		\right\}
		$\\
		$
		\delta'' =
		\begin{array}{l}
		\left\{
		q' \stackrel{[}{\longrightarrow} q'' \in \delta \ \vrule\  q' \ne p_I
		\right\}
		\cup
		\left\{
		q' \stackrel{]}{\longrightarrow} q'' \in \delta  \ \vrule\  q'' \ne p_T
		\right\}
		\cup\\
		\left\{
		q' \stackrel{x}{\longrightarrow} p_I \in \delta \right\}
		\cup
		\left\{
		p_T \stackrel{x}{\longrightarrow} q' \in \delta \right\}.
		\end{array}
		$

        It is easy to see that $L(G_1) = \ltagged(\Phi) \cap R_1$. If we remove tags by taking 
	$G_2 = (V_N, \Sigma, \{ X \to \sigma(R_X) \mid$ $ X \to R_X \in P \}, 
	S)$, we see that $G_2 \in$ HOP$(k, \Phi)$ by construction, and $L(G_2) = \ltaggable(\Phi) \cap R_0$.
        By Cor.~\ref{Cor:BoolAlgebra}, $L(G_0) \cap L(G_2) = L(G_0) \cap R_0$ is in HOP$(k,\Phi)$. 
\end{proof}

%
%
%
%
%
%
%

\section{Related work and conclusion}\label{sect:RelatedWorkConcl}
Earlier attempts have been made to generalize the \emph{operator precedence}  model and other similar 
grammar models. We discuss some relevant works and  explain how they differ 
from the \emph{higher-order operator precedence} model.  

Floyd  himself proposed the \emph{bounded-context} grammars
\cite{DBLP:journals/cacm/Floyd64}, which use left and right contexts of bounded
length to localize the edges of the handle; unfortunately,  the contexts
contain also nonterminals and so lose the closure properties of OP languages as
well as the possibility to do local parsing. 
\par
\emph{Chain-driven} languages 
\cite{CLMP17} are a recent extension of OP languages, which shares with HOP the idea of specifying the syntax 
structure by non-conflictual tags, but differs in technical ways we cannot 
describe here. The resulting family offers some  significant gain in expressive 
capacity over OP, enjoys local parsability, but it has poor closure properties, 
and cannot be easily formulated for  contexts larger than one terminal. 
Notice that the automata-theoretic approach presented in \cite{CLMP17} can be naturally applied to HOP 
languages for proving their local parsability.

	Since HOP extend the OP language family, which in turn include the input-driven 
	(or VP) language~\cite{CrespiMandrioli12} family,
	it is interesting to compare the HOP family with the recent extension of VP 
	languages, recognized by {\em tinput-driven pushdown automata} (TDPDA)
	\cite{Tinput-PDA-2015}, which enjoy similar closure properties. The families 
	HOP and TDPDA are incomparable: on one side,     the language  
	$\{a^n b a^n \mid n \geq 1\} \in \text{TDPDA} - \text{HOP}$, on the other side, \text{TDPDA} only recognize real-time languages, and thus fail the  non-realtime language which is 
	$\{a^m b^n c^n d^m \mid n, m \ge 1\} \cup \{a^m b^+ e d^m \mid m \ge 1\} 
	\in \text{HOP(3)}$.
Moreover the tinput parser is not suitable for local parsing, because it must 
operate from left to right, starting from the first character. 
 
\par
Recalling that OP grammars  have been applied in early grammar inference studies, we mention two loosely related  language classes  motivated by grammar 
inference research, which strives to discover expressive grammar types having 
good learnability properties. Within the so-called distributional approach, 
several authors have introduced various grammar types based on a common idea:  
that the syntax class of a word $v$ is determined by the left and right 
contexts of occurrence,  the context lengths  being finite integers $k$ and 
$\ell$. Two examples are: the $(k,\ell)$ \emph{substitutable} CF languages 
\cite{DBLP:conf/icgi/Yoshinaka08} characterized by the implication
$x_1 v y_1 u z_1 ,\, x_1 v y_2 u z_1 ,\, x_2 v y_1 u z_2  \in L$ $\text{ implies } x_2 v y_2 u z_2 \in L$ where  $|v|=k$ and $|u|= \ell$; 
and the related hierarchies of languages studied in \cite{DBLP:conf/icgi/LuqueL10}. A closer comparison of HOP and  language classes motivated  by grammar inference would be interesting. 
\medskip
\par
Since HOP is a new language model, its properties have been only partially studied. Thus, it remains to be seen whether other known theoretical properties of OP languages (such as the closure under concatenation and star or the invariance with respect to the CF non-counting property \cite{CreGuiMan81}) continue to hold for HOP. 
\par
We finish by discussing the potential for applications. First, the enhanced generative  capacity of higher degree HOP grammars in comparison to OP grammars may in principle  ease the task of writing syntactic specifications, but, of course, this needs to be evaluated for realistic cases.
We are confident that the practical parallel parsing algorithm in  \cite{BarenghiEtAl2015} can be extended from OP to HOP grammars.
\par
To apply  HOP to model-checking of infinite-state systems, the model has to be extended to $\omega$-languages and logically characterized, as recently done for OP languages in~\cite{LonatiEtAl2015}.
\par
Last, for grammar inference: we observe that it would be possible to define a partial order based on language inclusion, within each subfamily of HOP($k$) languages closed under Boolean operation, i.e., structurally compatible. 
Such a partially ordered set of grammars and languages, having the max-grammar as top element, is already known \cite{Crespi-ReghizziCACM73,Crespi-ReghizziMM1978} for the OP case,  and its lattice-theoretical properties  have been exploited for inferring grammars using just positive information sequences \cite{journals/csur/AngluinS83}. The availability of the $k$-ordered hierarchy may then enrich the learnable grammar space.
\bibliography{Floydbib}

\end{document}